\documentclass[preprint,12pt]{elsarticle}

\usepackage{amssymb}
\usepackage{amsmath}
\usepackage{amsthm}
\usepackage{color}

\newtheorem{theorem}{Theorem}
\newtheorem{proposition}{Proposition}
\newtheorem{corollary}{Corollary}

\begin{document}

\begin{frontmatter}

\title{Spherical Geometrical Bases of Spherical Origami}

\author{Takashi Yoshino}

\affiliation{organization={Department of Mechanical Engineering, Toyo University},
            addressline={Kujirai 2100},
            city={Kawagoe},
            postcode={350-8585},
            state={Saitama},
            country={Japan}}

\begin{abstract}
This paper establishes a rigorous geometrical framework for spherical origami---origami using spherical sheets based on spherical geometry---in two settings: origami restricted to the unit sphere $\mathbb{S}^2$, and three-dimensional folding of spherical sheets in space.
Three main results are proved.
First, all seven Huzita--Justin axioms admit explicit constructions on $\mathbb{S}^2$, so spherical origami inherits the same one-fold axiom structure as Euclidean origami.
Second, the Kawasaki--Justin theorem on flat-foldability extends to $\mathbb{S}^2$ by a local-tangent-plane argument, and the resulting vertex condition is equivalent to the angle-sum condition that characterizes Robertson's spherical $f$-tilings.
Third, taking equidistant curves (small arcs) as fold curves yields a one-parameter family of isometric 3D folds, parametrized by the dihedral angle whose closed-form expression is derived, and recovers geodesic folds as the flat limit.
The framework is illustrated by constructing computer graphics of spherical origami birds as a concrete example, supporting the practical utility of the proposed approach.
\end{abstract}

\begin{keyword}
spherical origami\sep spherical geometry\sep origami bird
\end{keyword}

\end{frontmatter}


\section{Introduction}\label{sec1}

Theories and applications of origami have progressed recently due to the increase in academic and technological interest.
One reason origami attracts mathematicians is that its manipulation (paper folding) differs from Euclidean geometry (using a compass and straightedge).
For example, origami provides an operation to divide an arbitrary angle equally into three, which is impossible in Euclidean geometry~\cite{hull2012project,geometricOrigami}.
On the other hand, the origami concept is now applied to various designs~\cite{origamiEngineering,liu2021origami}. Miura-ori~\cite{miuraoriReview} is a representative origami technology designed for space antennae, demonstrating that origami design is effective for such expanding structures.

A natural extension to consider is origami having curved faces.
There are two ways to realize curved faces: using flat paper and curved fold lines \cite{mitani2019curved,lukasheva2021curved}, and using curved paper directly.
Although the former method has been actively developed, the latter has seldom been explored because curved sheets are not sufficiently produced or made available.
The present article considers ``spherical origami"~\cite{kawasakiSphericalOrigami}, origami using spherical sheets, as the simplest case of the latter method.
Although spherical sheets are not yet widely available as a physical material, the mathematical framework developed here has potential relevance in several contexts.
In the design of deployable curved shell structures, such as retractable domes and curved antenna panels, fold-line geometry on curved surfaces must be specified precisely; the axioms and fold equations derived here provide the necessary computational tools for such tasks.
In computer graphics, accurate simulation of non-Euclidean surface deformations---for example, the rendering of spherical or dome-shaped flexible panels---requires explicit fold equations of the kind developed in this paper.
Finally, the present work contributes to the theoretical development of non-Euclidean origami: the framework runs in parallel with hyperbolic origami~\cite{hyperbolicOrigami}, and comparing the two geometries may yield general principles applicable to origami on surfaces of arbitrary curvature.

The objective of this study is to establish a geometrical basis for spherical origami on $\mathbb{S}^2$ and in three-dimensional (3D) space, grounded in spherical (non-Euclidean) geometry.
Specifically, the origami axioms are extended to the spherical setting and their validity is demonstrated through a practical case of producing computer graphics of spherical origami birds.
The definitions and axioms have been considered for constructing a mathematical formalism of Euclidean origami. Alperin and Lang~\cite{alperinAxioms} summarized the history of Euclidean origami and the contributions by Justin~\cite{justin}, Huzita~\cite{huzitaAxioms}, Hatori~\cite{Hatori}, and Lang~\cite{LangHP} toward the construction of the axioms.
They demonstrated that seven axioms of one fold are possible for Euclidean origami, based on eight definitions of Euclidean geometry.
Modifying the definitions and axioms in a similar way to apply them to spherical origami seems fruitful for comparing Euclidean and spherical origami.

Some researchers have discussed origami theory for curved paper using non-Euclidean geometry (e.g., Kawasaki~\cite{kawasakiSphericalOrigami,kawasakiBook} for spherical origami and Alperin et al.~\cite{hyperbolicOrigami} for hyperbolic origami).
Kawasaki studied flat origami of the 2-sphere~\cite{kawasaki1989,kawasaki1991}, listed the eight operations of spherical origami construction and proposed that they be summarized as four fundamental operations~\cite{kawasakiFour2011,kawasakiSphericalOrigami}, and showed that the spherical counterpart of the operation placing two given points onto two given lines simultaneously (Axiom~6 below) always has at least two solutions~\cite{kawasakiCount2012}.
In a book written for folders~\cite{kawasakiBook}, he treated a hemisphere as a spherical square and drew the crease pattern of a spherical bird base in great arcs.
These works concern folds along great circles, which keep the paper on the sphere.
To the author's knowledge, the equations of fold curves for folding part of a spherical sheet off its sphere, which are needed to draw three-dimensional spherical origami such as a crane, have not been given.
Thus, systematic definitions of spherical origami and practical results based on those definitions are necessary to better understand its features.

A related line of research concerns isometric foldings of Riemannian manifolds in the sense of Robertson~\cite{robertson1977}, who showed that the singularity set of such folds on surfaces forms an edge-to-edge tiling ($f$-tiling) of even vertex valency in which the alternating angle sums at each vertex equal $\pi$.
Subsequent work by Breda and collaborators classified and studied deformations of spherical $f$-tilings~\cite{breda2010}.
While this body of work also addresses folding on $\mathbb{S}^2$, its focus is the topological and combinatorial classification of tiling patterns arising from isometric folds, rather than the computational tools for constructing fold curves and fold operations required for origami design and computer graphics---the objective of the present work.

The structure of the paper and its main results are as follows.
First, the definitions for Euclidean origami are modified to match the constraints of spherical geometry, and Theorem~\ref{thm:SphericalHJ} establishes that all seven Huzita--Justin axioms admit explicit constructions on $\mathbb{S}^2$.
Second, Theorem~\ref{thm:SphericalKJ} establishes the spherical analogue of the Kawasaki--Justin theorem for flat-foldability on $\mathbb{S}^2$, and Corollary~\ref{cor:Ftiling} identifies its vertex condition with that of Robertson's spherical $f$-tilings.
Third, folding operations in 3D space are introduced using equidistant curves as fold curves, leading to Theorem~\ref{thm:3DFold}, which characterizes the resulting one-parameter family of isometric 3D folds and the equivalence of two transformation formulas, and Proposition~\ref{prop:DihedralAngle}, which gives a closed-form expression for the dihedral angle.
Finally, as a practical case, spherical origami birds are constructed in computer graphics as a concrete example of the framework.

\section{Origami on $\mathbb{S}^2$}

\subsection{Preliminary Settings}

\begin{figure}[ht]
\centering
\includegraphics[width=0.5\textwidth]{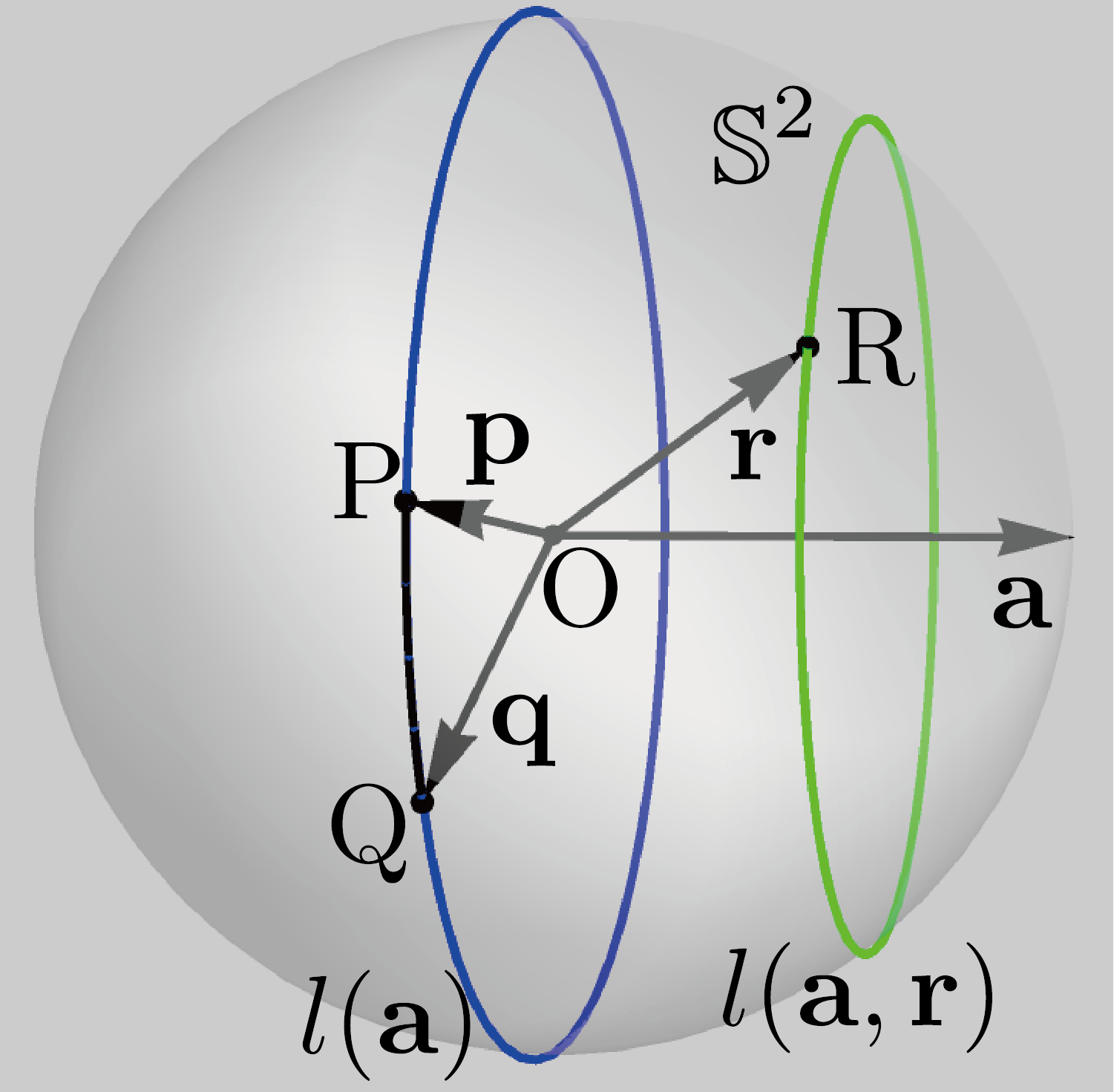}
\caption{Examples of points, a great circle and its pole, and a small circle (equidistant curve). }\label{Fig:Notation}
\end{figure}

Based on the general treatments of spherical geometry~\cite{toth}, the following definitions and representations are used in this paper.
We consider the surface of a unit sphere $\mathbb{S}^2$.
The center of the unit sphere is set at the origin of the Cartesian coordinates in the following sections except where specifically stated otherwise. Figure~\ref{Fig:Notation} shows examples of the notation. Points are represented by capital Latin letters (e.g., $\mathrm{P}$ and $\mathrm{Q}$), and the associated position vectors are denoted by the corresponding lowercase bold letters (e.g., $\mathbf{p}$ for $\mathrm{P}$).
Geodesic lines (great circles), which correspond to straight lines in Euclidean geometry, are denoted using the small italic letter $l$ and its pole ($\mathbf{a}$) as $l(\mathbf{a})$, or using simply the pole itself. Thus, any point $\mathbf{p}$ on the geodesic line represented by $\mathbf{a}$ satisfies $\mathbf{a}\cdot\mathbf{p}=0$. Poles are always set as unit vectors, although arbitrary vectors perpendicular to the plane containing the great circle are acceptable as its pole.
Equidistant curves (small circles), characteristic of non-Euclidean geometry, are denoted using the letter $l$, its pole $\mathbf{a}$, and any point on the circle $\mathbf{r}$ as $l(\mathbf{a}, \mathbf{r})$, or using simply $(\mathbf{a}, \mathbf{r})$. Such a curve comprises points equidistant from the geodesic line $\mathbf{a}$; however, it is not itself a geodesic, as is well known in non-Euclidean geometry.

The so-called Rodrigues' rotation formula~\cite{ogawaSphericalGeometry} is useful for drawing spherical patterns. It represents the rotation of a point $\mathrm{P}$ to $\mathrm{P'}$ by angle $\theta$ with respect to axis $\mathbf{a}$ as
\[
\mathbf{p}' = \mathbf{p} \cos\theta + (\mathbf{a}\times\mathbf{p})\sin\theta+\mathbf{a}(\mathbf{a}\cdot\mathbf{p}) (1-\cos\theta),
\]
where the axis $\mathbf{a}$ is set to be a unit vector. The poles of great and small circles are defined as unit vectors because the above formula is used to draw them.
For example, the formula for the spherical midpoint $\mathrm{P}_\mathrm{m}$ of $\mathrm{P}_1$ and $\mathrm{P}_2$ (when $\mathrm{P}_1$ and $\mathrm{P}_2$ are not antipodal, i.e., $\mathbf{p}_1\cdot\mathbf{p}_2\neq -1$) is derived from the formula as
\begin{equation}
  \mathbf{p}_\mathrm{m} = \frac{1}{\sqrt{2} \sqrt{1+\mathbf{p}_1\cdot\mathbf{p}_2}} \left(\mathbf{p}_1+\mathbf{p}_2\right). \label{Eq:MidPoint}
\end{equation}
Therefore, the rotation operation in this study is represented as
\begin{equation}
    \mathrm{rot}(\mathbf{p}, \mathbf{a},\theta) := \mathbf{p} \cos\theta + (\mathbf{a}\times\mathbf{p})\sin\theta+\mathbf{a}(\mathbf{a}\cdot\mathbf{p}) (1-\cos\theta), \label{Eq:Rotation}
\end{equation}
where $\mathbf{p}$, $\mathbf{a}$, and $\theta$ represent the position vector before the rotation, the rotational axis, and the amount of the rotation, respectively.
Using this notation, a small circle $l(\mathbf{a}, \mathbf{r})$ is expressed as $\mathrm{rot}(\mathbf{r}, \mathbf{a}, t)\, (0\le t < 2\pi)$.

The description for obtaining the intersection of two curves should be considered for making computer graphics. Two cases are discussed here. One is the intersections of two great circles, and the other is those of a great circle and a small circle. The equation for the intersections \(\mathbf{q}\) of two great circles \(\mathbf{a}_1\) and \(\mathbf{a}_2\) is
\[
\mathbf{q} = \pm \frac{\mathbf{a}_1\times\mathbf{a}_2}{\sqrt{1-(\mathbf{a}_1\cdot\mathbf{a}_2)^2}}.
\]
These are in antipodal relation, which is clear because a spherical digon consists of two great circles. On the other hand, the intersections of great and small circles, $\mathbf{a}_1$ and ($\mathbf{a}_2$, $\mathbf{p}$), can only be obtained numerically by solving the equation
\[
\mathrm{rot}(\mathbf{p}, \mathbf{a}_2, t)\cdot \mathbf{a}_1 =0,
\]
for $t$ and substituting the solution into $\mathrm{rot}(\mathbf{p}, \mathbf{a}_2, t)$.

\subsection{Definitions of Reflection and Fold Curve}

To define a ``fold" on $\mathbb{S}^2$, the reflection operation on $\mathbb{S}^2$ must be clarified.
``Reflection of $\mathrm{P}$ with respect to the geodesic line $\mathbf{a}$" is a plane-symmetrical transformation of $\mathrm{P}$ with respect to the plane containing the great circle $\mathbf{a}$.
The operation is expressed as
\begin{equation}
    \mathrm{ref}({\mathbf{p}, \mathbf{a}}):=\mathbf{p} -2 (\mathbf{a}\cdot\mathbf{p})\mathbf{a}. \label{Eq:Reflection}
\end{equation}
Figure~\ref{Fig:Reflection} shows an example of reflection.
A spherical point $\mathrm{R}$ transforms to the point $\mathrm{R'}$ with respect to the plane perpendicular to the unit vector $\mathbf{a}_{\mathrm{PQ}}$.
The reflection operation always transfers a spherical point onto the spherical surface without normalization because the reflection plane contains the sphere's center.

\begin{figure}[ht]
\centering
\includegraphics[width=0.5\textwidth]{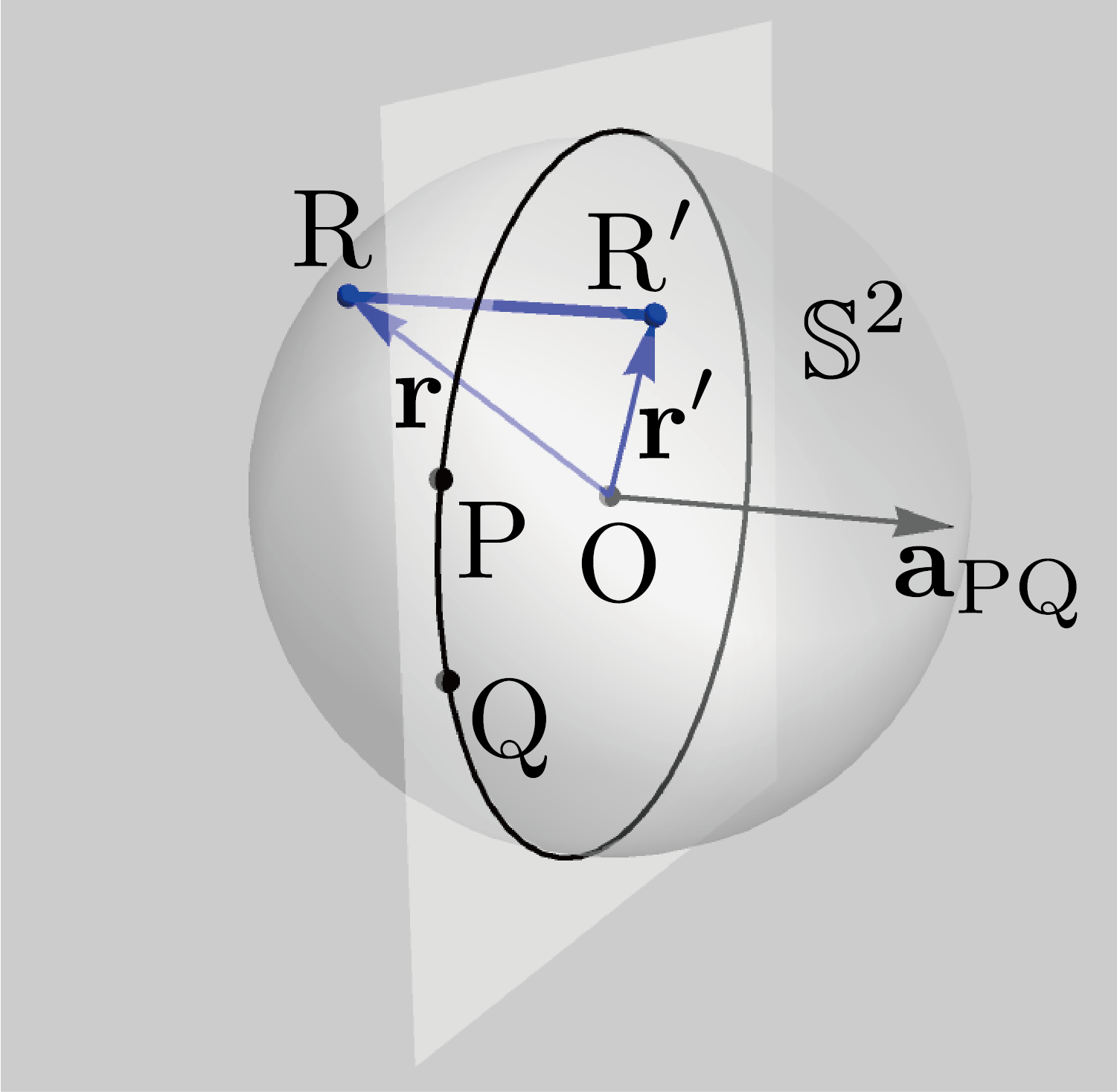}
\caption{Examples of a reflection and a fold curve on $\mathbb{S}^2$.}\label{Fig:Reflection}
\end{figure}

Geodesic lines represent fold curves in most cases.
The correspondence of spherical geometry with Euclidean geometry indicates that the fold curves (not ``fold lines", in a strict sense) on $\mathbb{S}^2$ are great circles because great circles are the concept corresponding to lines in Euclidean origami.
In the following, great circles are chosen as fold curves of origami on $\mathbb{S}^2$.
Let $\mathrm{P}$, $\mathrm{Q}$, and $\mathrm{R}$ be different spherical points and $\mathrm{O}$ be the center of the sphere, as shown in Fig.~\ref{Fig:Reflection}.
The geodesic line from $\mathrm{P}$ to $\mathrm{Q}$, corresponding to the fold curve, is defined by the pole $\mathbf{a}_{\mathrm{PQ}} = (\mathbf{p}\times \mathbf{q})/\left|\mathbf{p}\times \mathbf{q}\right|$.
The geodesic line (great arc) $\mathrm{PQ}$ is written using $\mathbf{a}_{\mathrm{PQ}}$ as
\[
\mathrm{rot}(\mathbf{p},\mathbf{a}_{\mathrm{PQ}},t)\, (0\le t \le \cos^{-1} (\mathbf{p}\cdot\mathbf{q})).
\]
Thus, the point $\mathrm{R}$ in Fig.~\ref{Fig:Reflection} is transferred to $\mathrm{R'}$ by the folding with $\mathbf{a}_{\mathrm{PQ}}$ according to Eq.~\ref{Eq:Reflection}.

\begin{figure}[ht]
\centering
\includegraphics[width=0.9\textwidth]{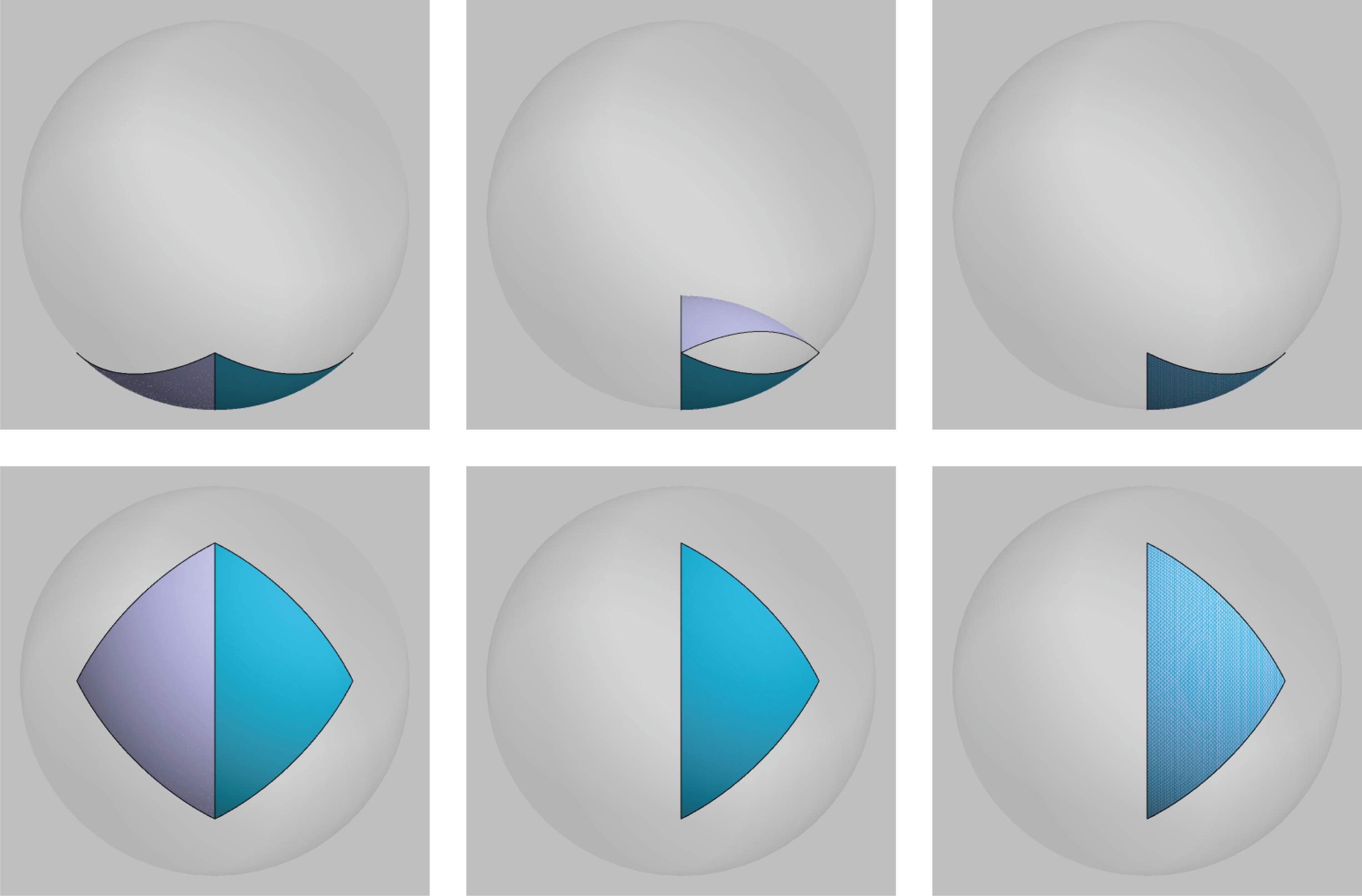}
\caption{Folding of a rhombic spherical sheet. Top and bottom images give top and side views, respectively. An earlier version of this figure appeared in~\cite{yoshinoJCDCG2024}. }\label{Fig:RotSphTriangle}
\end{figure}

Based on the definitions of a fold curve, folding on $\mathbb{S}^2$ is described as the reflection of a spherical polygon along the fold curve.
A fold on $\mathbb{S}^2$ is not a simple rotation, unlike in Euclidean origami.
Because rotation is carried out with an axis and the axis is a straight line in the 3D space, the line containing the terminal points of a fold curve (great arc) is taken as the axis for rotation.
Figure~\ref{Fig:RotSphTriangle} illustrates a folding operation of a piece of spherical sheet:
The left panels show rhombic spherical sheets divided into two congruent triangles by a great arc.
Our interest is folding the left triangle along the great arc, and the triangles must entirely overlap at the end of the operation.
The middle panels depict the rotation of the left triangle around the axis containing the two common vertices until the rotated vertex contacts the spherical surface or the point reflected across the plane containing the great arc.
The top view (top middle panel) suggests that folding is not equivalent to rotation.
The right panels display the result of the reflection of the triangle across the plane containing the three vertices of the rotated triangle.
As shown, the two triangles overlap entirely and the left triangle has been folded successfully.
This result shows that folding on $\mathbb{S}^2$ necessarily involves a reflection operation.

\subsection{Definitions and Axioms for Spherical Origami} \label{subsec:DefAndAxioms}

This subsection considers the axioms for spherical origami after introducing the definitions. Its objective is to show that origami on $\mathbb{S}^2$ can be constructed using spherical geometry.
Based on the descriptions by Alperin and Lang~\cite{alperinAxioms} listed in Appendix~\ref{Appendix:PlanarOrigami}, the definitions and axioms for spherical geometry are described in one-to-one correspondence.
First, the modified definitions are listed.
Then, the equation or implementation of spherical geometry is presented for each axiom.

The main changes in the definitions are the use of 3D Cartesian coordinates and the use of the pole of a great circle to represent a fold curve.
Recall that the pole of a great circle is always normalized to a unit vector because the definition of the rotation Eq.~\ref{Eq:Rotation} requires using a unit vector as its axis.
Although it is possible to represent arbitrary points using two variables by taking the spherical coordinates, Cartesian coordinates are used in the present study because the extension from Euclidean origami is straightforward. The resulting definitions are listed as follows:

\begin{enumerate}
    \item A point $\mathbf{p}=(x, y, z)$ has the ordered Cartesian coordinates $x$, $y$, and $z$ satisfying $|\mathbf{p}|^2 = x^2+y^2+z^2=1$.
    \item  A geodesic line $\mathbf{a}=(X, Y, Z)$ is the set of points $\mathbf{p}=(x, y, z)$ that satisfy the equations $\mathbf{a}\cdot\mathbf{p}=Xx+Yy+Zz=0$ and $|\mathbf{a}|^2 =X^2+Y^2+Z^2=1$.
    \item The \textit{folded image} $F_{L_\mathrm{F}}(\mathbf{p})$ of a spherical point $\mathbf{p}=(x, y, z)$ across a fold curve $\mathbf{a}_\mathrm{F}=(X_\mathrm{F}, Y_\mathrm{F}, Z_\mathrm{F})$ is the reflection across the fold curve.
    \item The \textit{folded image} $F_{L_\mathrm{F}}(l(\mathbf{a}))$ of a line $\mathbf{a}=(X, Y, Z)$ across a fold line $l(\mathbf{a}_\mathrm{F})$ where $\mathbf{a}_\mathrm{F}=(X_\mathrm{F}, Y_\mathrm{F}, Z_\mathrm{F})$ is the reflection of the line across the fold line.
    \item Given two points $\mathbf{p}_1=(x_1, y_1, z_1)$ and $\mathbf{p}_2=(x_2, y_2, z_2)$, the alignment $\mathrm{P}_1\leftrightarrow \mathrm{P}_2$ is satisfied if and only if $x_1=x_2$, $y_1=y_2$, and $z_1=z_2$.
    \item Given two lines $\mathbf{a}_1=(X_1, Y_1, Z_1)$ and $\mathbf{a}_2=(X_2, Y_2, Z_2)$, the alignment $\mathbf{a}_1\leftrightarrow \mathbf{a}_2$ is satisfied if and only if either $X_1=X_2$, $Y_1=Y_2$, and $Z_1=Z_2$ or $X_1=-X_2$, $Y_1=-Y_2$, and $Z_1=-Z_2$.
    \item Given a point $\mathbf{p}=(x, y, z)$ and a line $\mathbf{a}=(X, Y, Z)$, the alignment $\mathbf{p}\leftrightarrow \mathbf{a}$ is satisfied if and only if $\mathbf{p}\cdot\mathbf{a}= Xx+Yy+Zz=0$.
    \item A one-fold axiom (1FA) is a minimal set of alignments that define a single fold curve on a finite region of the spherical surface with a finite number of solutions.
\end{enumerate}

Next, the origami axioms and their implementations for spherical origami are described sequentially.
These axioms are descriptions of the equations of the fold curves satisfying the requirements.

\begin{theorem}[Spherical Huzita--Justin Axioms]\label{thm:SphericalHJ}
Under the spherical definitions 1--8 above, every Huzita--Justin axiom (Axioms~1--7 of Appendix~\ref{subsec:HuzitaJustin}) admits an explicit construction on $\mathbb{S}^2$. The pole $\mathbf{a}$ of the fold curve realizing the prescribed alignments is given by:
\begin{itemize}
\item a closed-form algebraic expression in the inputs for Axioms~1--4;
\item an explicit system of equations (solvable numerically) for Axioms~5--7.
\end{itemize}
The solution counts are as follows: Axioms~1, 2, and 4 each have a unique solution; Axiom~3 has two solutions (the bisector and external bisector); Axiom~5 admits zero, one, or two solutions depending on the inputs; Axiom~6 admits finitely many solutions depending on the inputs; and Axiom~7 has a unique solution when $\mathbf{a}_1$ and $\mathbf{a}_2$ are not antipodal.
Consequently, spherical origami possesses the same seven-axiom structure as Euclidean origami, and any fold constructible in Euclidean origami has a spherical counterpart.
\end{theorem}

\begin{proof}
We exhibit the fold-pole construction for each axiom in turn.

Axiom 1 states, ``Given two points $\mathrm{P}_1$ and $\mathrm{P}_2$, we can fold a line connecting them." When $\mathrm{P}_1$ and $\mathrm{P}_2$ are not antipodal, the fold curve is given by the pole of the great circle $\mathbf{a}$ as
\[
\mathbf{a}=\frac{\mathbf{p}_1\times\mathbf{p}_2 }{\left| \mathbf{p}_1\times\mathbf{p}_2\right|}.
\]
The pole is unique up to sign ($\pm\mathbf{a}$ represent the same great circle), so the fold curve is uniquely determined.

Axiom 2 states, ``Given two points $\mathrm{P}_1$ and $\mathrm{P}_2$, we can fold $\mathrm{P}_1$ onto $\mathrm{P}_2$." The fold curve $\mathbf{a}$ is expressed as
\[
\mathbf{a}= \frac{\mathbf{p}_{\mathrm{m}}\times\mathbf{a}_{12} }{\left| \mathbf{p}_{\mathrm{m}}\times\mathbf{a}_{12}\right|},
\]
where $\mathbf{a}_{12}$ is
\[
\mathbf{a}_{12}=\frac{\mathbf{p}_1\times\mathbf{p}_2 }{ \left| \mathbf{p}_1\times\mathbf{p}_2\right|},
\]
and $\mathbf{p}_{\mathrm{m}}$ is the midpoint of $\mathrm{P}_1$ and $\mathrm{P}_2$ calculated from Eq.~\ref{Eq:MidPoint}.

Axiom 3 states, ``Given two lines $\mathbf{a}_1$ and $\mathbf{a}_2$, we can fold line $\mathbf{a}_1$ onto $\mathbf{a}_2$."
As in the Euclidean case, there are in general two solutions corresponding to the two angle bisectors of $\mathbf{a}_1$ and $\mathbf{a}_2$.
The poles of the two fold curves are
\[
\mathbf{a} = \frac{1}{\sqrt{2} \sqrt{1+\mathbf{a}_1\cdot\mathbf{a}_2}} \left(\mathbf{a}_1+\mathbf{a}_2\right)
\quad \text{and} \quad
\mathbf{a} = \frac{1}{\sqrt{2} \sqrt{1-\mathbf{a}_1\cdot\mathbf{a}_2}} \left(\mathbf{a}_1-\mathbf{a}_2\right),
\]
where the first solution is defined when $\mathbf{a}_1\cdot\mathbf{a}_2\neq -1$ and the second when $\mathbf{a}_1\cdot\mathbf{a}_2\neq 1$ (i.e., when $\mathbf{a}_1\neq\pm\mathbf{a}_2$).
These two fold curves are perpendicular to each other and correspond to the bisector and external bisector of the angle between the two great circles.

Axiom 4 states, ``Given a point $\mathrm{P}_1$ and a line $l(\mathbf{a}_1)$, we can make a fold perpendicular to $l(\mathbf{a}_1)$ passing through the point $\mathrm{P}_1$."
The pole of the fold curve $\mathbf{a}$ is
\[
\mathbf{a}=\frac{\mathbf{p}_1\times\mathbf{a}_1 }{ \left| \mathbf{p}_1\times\mathbf{a}_1\right|}.
\]

Axiom 5 states, ``Given two points $\mathrm{P}_1$ and $\mathrm{P}_2$ and a line $\mathbf{a}_1$, we can make a fold that places $\mathrm{P}_1$ onto $\mathbf{a}_1$ and passes through the point $\mathrm{P}_2$."
A solution exists if and only if the small circle $l(\mathbf{p}_2, \mathbf{p}_1)$ (the circle through $\mathrm{P}_1$ having pole $\mathrm{P}_2$) intersects the great circle $\mathbf{a}_1$; this fails when the spherical distance from $\mathrm{P}_1$ to $\mathbf{a}_1$ exceeds the radius of $l(\mathbf{p}_2, \mathbf{p}_1)$.
When two intersections exist, each gives a distinct fold curve; when the circle is tangent to $\mathbf{a}_1$, there is exactly one solution.
Let $\mathrm{P}_1'$ be one of the intersections of $l(\mathbf{p}_2, \mathbf{p}_1)$ and $\mathbf{a}_1$.
Then, the pole of the corresponding fold curve $\mathbf{a}$ is
\[
\mathbf{a}=\frac{\mathbf{p}_{\mathrm{m}}\times\mathbf{p}_{2} }{\left| \mathbf{p}_{\mathrm{m}}\times\mathbf{p}_{2}\right|} ,
\]
where $\mathbf{p}_{\mathrm{m}}$ is the position vector of the midpoint of $\mathrm{P}_1$ and $\mathrm{P}_1'$.

Axiom 6 states, ``Given two points $\mathrm{P}_1$ and $\mathrm{P}_2$ and two lines $\mathbf{a}_1$ and $\mathbf{a}_2$, we can make a fold that places $\mathrm{P}_1$ onto line $\mathbf{a}_1$ and $\mathrm{P}_2$ onto line $\mathbf{a}_2$."
Let one of the intersections of $\mathbf{a}_1$ and $\mathbf{a}_2$ be $\mathrm{Q}$ and define two transferred points after the fold as $\mathrm{P}_1'$ on $\mathbf{a}_1$ and  $\mathrm{P}_2'$ on $\mathbf{a}_2$.
The spherical distances from $\mathrm{Q}$ to $\mathrm{P}_1'$ and $\mathrm{P}_2'$ are set to be $t_1$ and $t_2$ satisfying the requirements. Then, the position vectors $\mathbf{p}_1'$ and  $\mathbf{p}_2'$ are expressed as
$\mathbf{p}_1' = \mathrm{rot}(\mathbf{q}, \mathbf{a}_1, t_1)$ and $\mathbf{p}_2' = \mathrm{rot}(\mathbf{q}, \mathbf{a}_2, t_2)$, respectively.
The two midpoints $\mathrm{P}_\mathrm{1m}$ and $\mathrm{P}_\mathrm{2m}$ are set as those of $\mathrm{P}_1$ and $\mathrm{P}_1'$ and $\mathrm{P}_2$ and $\mathrm{P}_2'$, respectively. Then, the values of $t_1$ and $t_2$ are obtained by solving simultaneous equations $\mathbf{a}_1\cdot(\mathbf{p}_\mathrm{1m} \times \mathbf{p}_\mathrm{2m})=0$ and $\mathbf{a}_2\cdot(\mathbf{p}_\mathrm{1m} \times \mathbf{p}_\mathrm{2m})=0$. Finally, the pole of the fold curve $\mathbf{a}$ is obtained as
\[
\mathbf{a}=\frac{\mathbf{p}_\mathrm{1m} \times \mathbf{p}_\mathrm{2m}}{\left| \mathbf{p}_\mathrm{1m} \times \mathbf{p}_\mathrm{2m}\right|}
\]
by substituting in the solutions $t_1$ and $t_2$ of the simultaneous equations.
As in the Euclidean case, where the analogous constraint reduces to a cubic equation yielding up to three solutions~\cite{alperinAxioms}, the spherical simultaneous equations are also nonlinear and may admit multiple solutions; in the present implementation these are obtained numerically.
The number of solutions on the sphere was studied by Kawasaki~\cite{kawasakiCount2012}, using the curve formed by the midpoints of the great arcs that join a given point to the points of a given great circle. He showed that the operation always has at least two solutions, and gave configurations with two, three, and four. Unlike the Euclidean case, in which the corresponding fold may fail to exist, the spherical Axiom~6 therefore never fails; whether four is the maximum is, to the author's knowledge, open.

Axiom 7 states, ``Given a point $\mathrm{P}_1$ and two lines $\mathbf{a}_1$ and $\mathbf{a}_2$, we can make a fold perpendicular to $\mathbf{a}_2$ that places $\mathrm{P}_1$ onto line $\mathbf{a}_1$."
Let the intersection of $\mathbf{a}_1$ and $\mathbf{a}_2$ be $\mathrm{Q}$ and define a point $\mathrm{P}_1'$ on $\mathbf{a}_1$ such that the spherical distance from $\mathrm{Q}$ is $t$, so that $\mathbf{p}_1' = \mathrm{rot}(\mathbf{q}, \mathbf{a}_1, t)$.
By solving the equation $\mathbf{a}_2\cdot (\mathbf{p}_1\times\mathbf{p}_1')=0$, the value of $t$ can be obtained.
Point $\mathrm{P}_\mathrm{m}$, the midpoint of $\mathrm{P}_1$ and $\mathrm{P}_1'$, is determined by substituting the value of $t$ into $\mathbf{p}_1'$.
Finally, the pole of the fold curve $\mathbf{a}$ is obtained as
\[
\mathbf{a}=\frac{\mathbf{a}_2 \times \mathbf{p}_\mathrm{m}}{\left| \mathbf{a}_2 \times \mathbf{p}_\mathrm{m}\right|}.
\]
The pole is uniquely determined when $\mathbf{a}_1$ and $\mathbf{a}_2$ are not antipodal so that $\mathrm{Q}$ is well defined.
This completes the construction for all seven axioms.
\end{proof}

\subsection{Flat-Foldability on $\mathbb{S}^2$} \label{subsec:FlatFoldability}

Beyond the existence of the seven axioms, the one-fold structure on $\mathbb{S}^2$ inherits the Kawasaki--Justin flat-foldability condition from the Euclidean case. The proof relies on the fact that the local tangent plane at any vertex of $\mathbb{S}^2$ is Euclidean, so the standard reflection-composition argument transfers verbatim. Flat origami of the 2-sphere was studied by Kawasaki in the late 1980s~\cite{kawasaki1989,kawasaki1991}.

\begin{theorem}[Spherical Kawasaki--Justin]\label{thm:SphericalKJ}
Let a single-vertex crease pattern on $\mathbb{S}^2$ consist of $n$ great-circle fold curves through a common vertex, with consecutive face angles $\alpha_1, \alpha_2, \ldots, \alpha_n$ summing to $2\pi$.
The pattern is flat-foldable (i.e., foldable entirely within $\mathbb{S}^2$) if and only if $n$ is even and the alternating sum of the face angles vanishes:
\[
\alpha_1 - \alpha_2 + \alpha_3 - \cdots - \alpha_n = 0.
\]
\end{theorem}

\begin{proof}
The argument is local at the vertex. Since spherical geometry coincides with Euclidean geometry on the tangent plane at the vertex, the Euclidean Kawasaki--Justin theorem~\cite{hull2012project,alperinAxioms} transfers directly. We sketch the necessary direction explicitly.

A flat fold on $\mathbb{S}^2$ is a reflection across the great circle representing the fold curve (Eq.~\ref{Eq:Reflection}), and the composition of $n$ such reflections through a common vertex is an isometry fixing that vertex.
For the pattern to fold flat, this composite must be the identity, which holds if and only if $n$ is even and the resulting rotation angle---twice the alternating sum of the angles between consecutive fold lines---vanishes modulo $2\pi$.
The converse---that the angle condition implies the existence of a flat-folded state---is classical~\cite{hull2012project} and rests on the same local Euclidean structure at the vertex.
\end{proof}

\begin{corollary}[Equivalence with Robertson's $f$-tilings]\label{cor:Ftiling}
A flat-foldable single-vertex crease pattern on $\mathbb{S}^2$ in the sense of Theorem~\ref{thm:SphericalKJ} arises as the local vertex structure of a Robertson spherical $f$-tiling~\cite{robertson1977}, where the singularity set of an isometric folding satisfies the alternating angle-sum condition $\alpha_1 + \alpha_3 + \cdots = \alpha_2 + \alpha_4 + \cdots = \pi$ at each vertex of even valency.
\end{corollary}

\begin{proof}
The condition $\sum_i \alpha_i = 2\pi$ together with the alternating sum equation of Theorem~\ref{thm:SphericalKJ} yields $\alpha_1 + \alpha_3 + \cdots = \alpha_2 + \alpha_4 + \cdots = \pi$, which is Robertson's condition for an $f$-tiling vertex. The converse follows by reversing the argument.
\end{proof}

Theorem~\ref{thm:SphericalKJ} guarantees that any Euclidean flat-foldable single-vertex pattern transfers to $\mathbb{S}^2$ without modification of the angle condition. The next subsection illustrates this with the concrete case of origami birds.

\subsection{Spherical Origami Birds on $\mathbb{S}^2$}

The folding of an origami bird on $\mathbb{S}^2$ is discussed in this subsection. The seven axioms of Section~\ref{subsec:DefAndAxioms} (Theorem~\ref{thm:SphericalHJ}) provide the construction primitives, while Theorem~\ref{thm:SphericalKJ} ensures that any single-vertex crease pattern satisfying the alternating angle condition can be folded flat on $\mathbb{S}^2$. Combined, these results allow direct translation of Euclidean origami constructions to $\mathbb{S}^2$.

\begin{figure}[ht]
\centering
\includegraphics[width=0.6\textwidth]{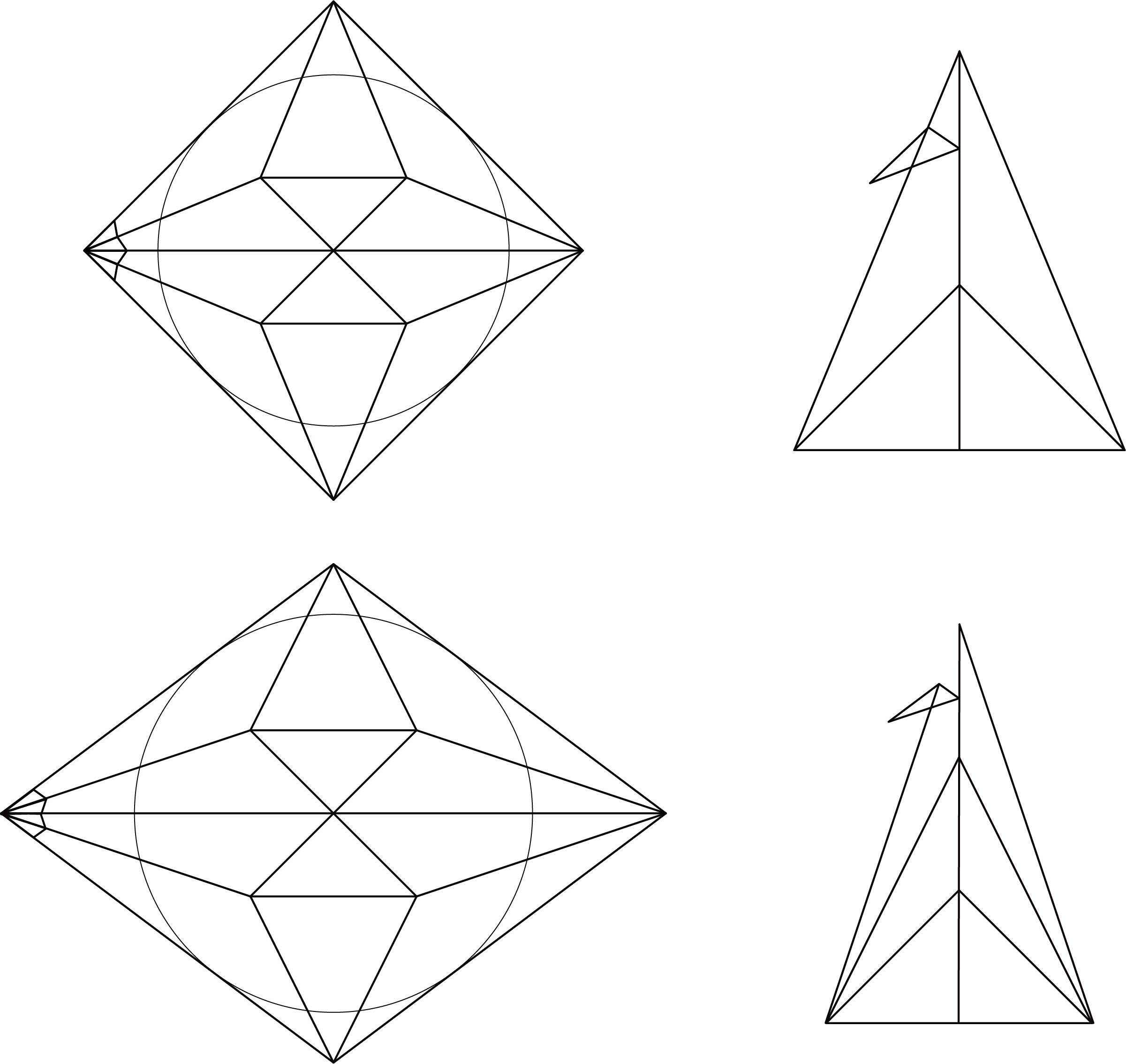}
\caption{Foldout diagrams and the folded forms before the expansions of wings for Euclidean origami birds. Top: a square sheet. Bottom: a rhombic sheet.}\label{Fig:FoldoutPlanarBird}
\end{figure}

For Euclidean origami, there is a sheet-level condition for origami birds, distinct from the single-vertex angle condition of Theorem~\ref{thm:SphericalKJ}. Justin~\cite{justinCrane} showed that a quadrilateral can be folded into a perfect bird base, the kind obtained from a square, if and only if the sums of its opposite sides are equal; for a convex quadrilateral this is equivalent to the existence of an inscribed circle. Kawasaki~\cite{kawasakiBook} discusses the same condition for folders, including non-convex sheets. All sheets considered here are convex, and we refer to the condition as the incircle condition.
Figure~\ref{Fig:FoldoutPlanarBird} illustrates two cases of foldout diagrams with inscribed circles (left) and the results of the folds before the expansions of the wings (right).
The drawings on the right represent the final output of the folding process on a plane; only the expansions of the wings to 3D space remain.
Although the expansion is trivial for Euclidean origami, it is not easy for spherical origami, as discussed in the following section.
The tips of the wings and the tail overlap in the case of a square sheet (bottom right) because the distances from the center of the sheet to the tips of the wings and tail are equal.
Hidden line removal was not applied to the images in the figure.

\begin{figure}[ht]
\centering
\includegraphics[width=0.95\textwidth]{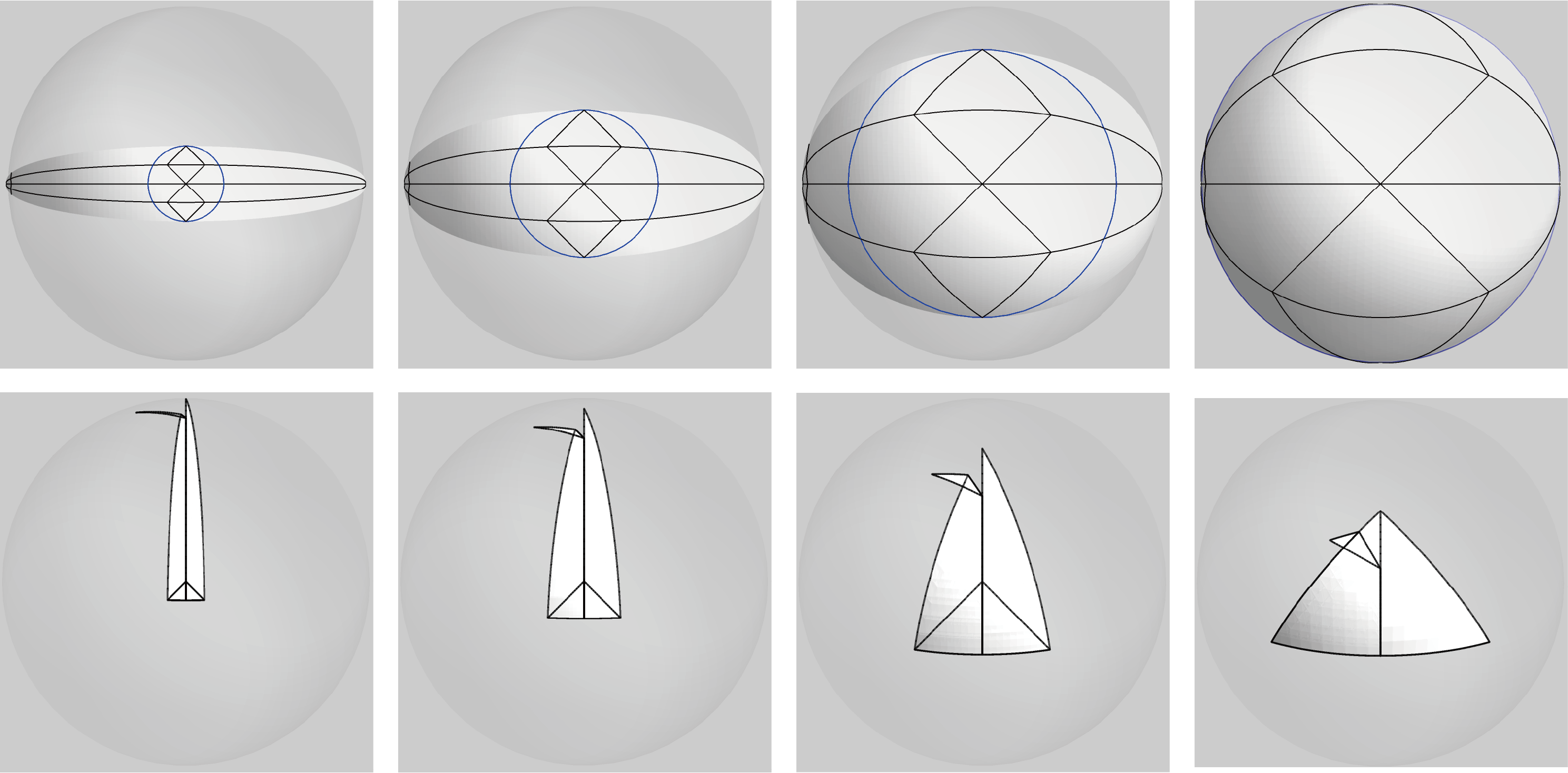}
\caption{Examples of foldout diagrams and folds on $\mathbb{S}^2$ for a hemisphere (right) and spherical digons (others). An earlier version of the foldout diagrams appeared in~\cite{yoshinoJCDCG2024}.}\label{Fig:FoldoutAndResultsInS2}
\end{figure}

Foldout diagrams of spherical origami birds can be drawn if the sheets are spherical digons.
Spherical digons whose inner angle is less than or equal to $\pi$ satisfy the incircle condition.
Figure~\ref{Fig:FoldoutAndResultsInS2} shows the foldout diagrams for origami birds.
The top row shows foldout diagrams for spherical digons and hemispherical sheets with their inscribed circles. The images were drawn based on the foldout diagrams for Euclidean origami in Fig.~\ref{Fig:FoldoutPlanarBird}.
The origami axioms confirm the drawing processes. For example, axiom 3 was used to divide angles equally into two.
The diagram is determined automatically by the inner angle value.
The differences among the four images are due to the amount of the inner angle, which from left to right is $\pi/8$, $\pi/4$, $\pi/2$,  and $\pi$.

The bottom row of Fig.~\ref{Fig:FoldoutAndResultsInS2} shows the final forms within $\mathbb{S}^2$. Each image was obtained from the above foldout diagram. The final folds are performed using the reflections (Eq.~\ref{Eq:Reflection}) across suitable fold curves and the fold order. The right-most form is folded from a hemispherical sheet.
The tips of the wings, tail, and back overlap completely in this case. The results indicate that the spherical origami on $\mathbb{S}^2$ works well, at least for origami birds.

\section{Spherical Origami in 3D} \label{sec:3D}

\subsection{New Definition of Fold Curve}

Because fold curves (great arcs) are applicable to only reflection transforms, as shown in Fig.~\ref{Fig:RotSphTriangle}, they must be replaced with new fold curves representing folds in 3D space.
A fold to the outside of $\mathbb{S}^2$ is referred to as a ``3D fold" in this paper.
The fold contains a reflection as folding on $\mathbb{S}^2$ always involves a reflection operation.
In addition, the new fold curves must include great circles as special cases.

Small arcs (equidistant curves) are taken as fold curves instead of great arcs.
The new fold curves must be the intersection of two spherical surfaces because the folded face does not change its shape.
Furthermore, great arcs are involved in the members of small circles.
Figure~\ref{Fig:NewFoldCurve} illustrates the overlap of two equal spherical surfaces, denoted as $\mathrm{O}$ and $\mathrm{O}_1$.
Based on Fig.~\ref{Fig:NewFoldCurve}, the fold of spherical triangle $\mathrm{PQR}$ on the spherical surface $\mathrm{O}$ with respect to the great arc $\mathrm{PQ}$ is considered in the following.
The point $\mathrm{O}$ is the center of the spherical surface containing the spherical points $\mathrm{P}$, $\mathrm{Q}$, and $\mathrm{R}$, and $\mathrm{O}_1$ is the center of the spherical surface containing the folded triangle $\mathrm{PQR'}$.
The viewpoint of Fig.~\ref{Fig:NewFoldCurve}A is taken so that the points $\mathrm{P}$ and $\mathrm{Q}$ overlap.
The point denoted as $\mathrm{P_m}$ in the figure is the midpoint of $\mathrm{P}$ and $\mathrm{Q}$ on the spherical surface $\mathrm{O}$.
Figure~\ref{Fig:NewFoldCurve}B is the view from the direction of $\mathbf{p}_\mathrm{m}$.

\begin{figure}[ht]
\centering
\includegraphics[width=0.95\textwidth]{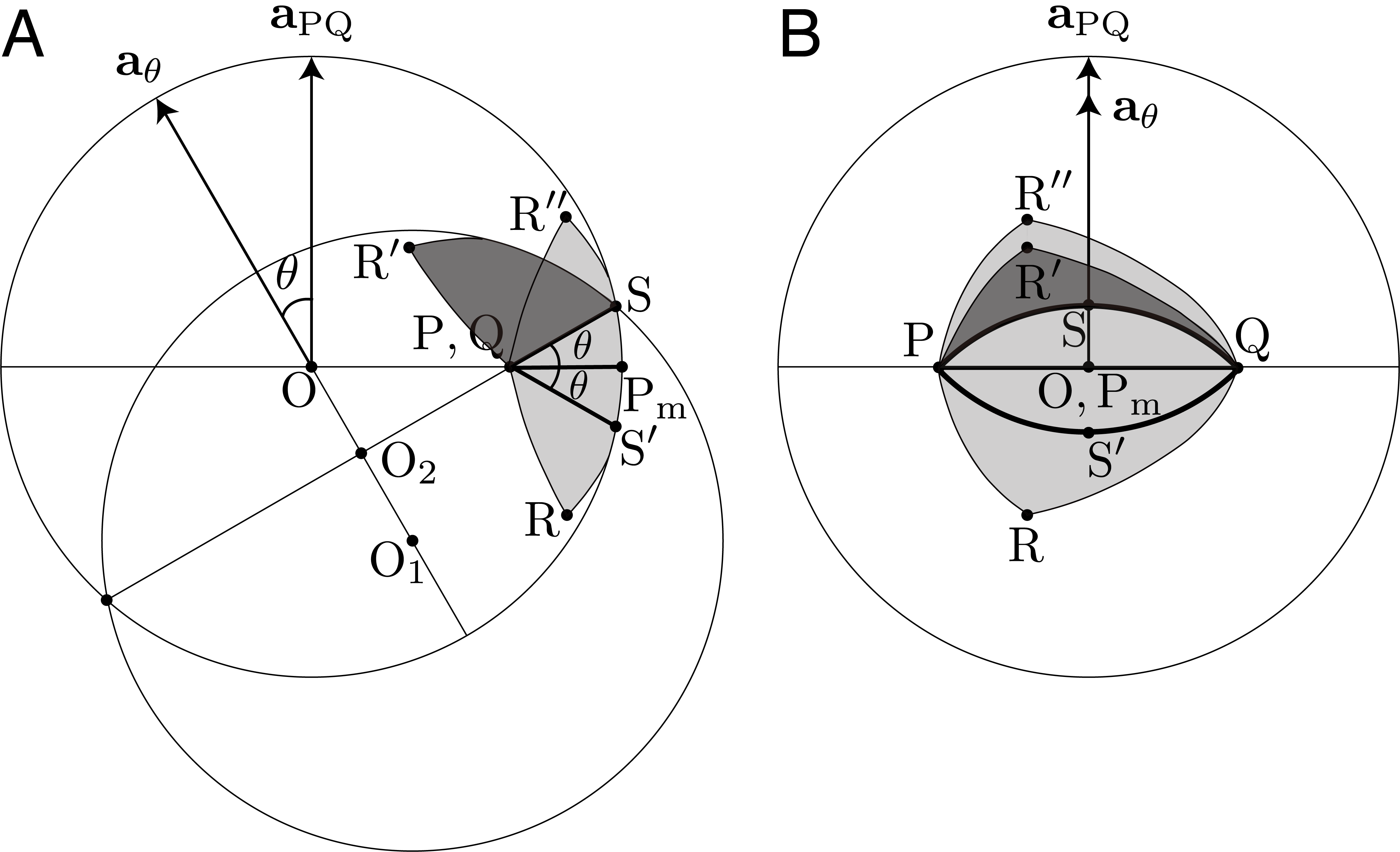}
\caption{Folding of the triangle $\mathrm{PQR}$ across the great arc $\mathrm{PQ}$ in 3D space. A: The viewpoint is set to overlap the vertices $\mathrm{P}$ and $\mathrm{Q}$. B: The viewpoint is set to overlap the points $\mathrm{O_1}$ and $\mathrm{P_m}$. }\label{Fig:NewFoldCurve}
\end{figure}

A new fold curve for the 3D folds is determined by the position vectors $\mathbf{p}$ and $\mathbf{q}$ and the value of an angle $\theta$, as shown in Fig.~\ref{Fig:NewFoldCurve}.
Here, $\theta$ is the angle between the pole of the original fold curve $\mathrm{PQ}$, $\mathbf{a}_\mathrm{PQ}=(\mathbf{p}\times \mathbf{q})/|\mathbf{p}\times \mathbf{q}|$, and the pole of the new fold curve (small arc) $\mathbf{a}_\theta$.
Therefore, the pole for the new fold curve can be expressed as
\[
\mathbf{a}_\theta=\mathrm{rot}(\mathbf{a}_\mathrm{PQ},\mathbf{a}_\mathrm{PQ}\times \mathbf{p}_\mathrm{m},\theta).
\]
The plane perpendicular to $\mathbf{a}_\theta$ and containing $\mathrm{P}$ and $\mathrm{Q}$ is automatically determined, and the intersection of the plane and the spherical surface $\mathrm{O}$ is the small circle containing the new fold curve.
The equation for the small circle $(\mathbf{a}_\theta,\mathbf{p})$ is written in terms of a parameter $t$ as
\[
 \mathrm{rot}(\mathbf{p}, \mathbf{a}_\theta, t) \,\,\,\, (0 \le t < 2\pi).
\]
Because the new fold curve starts from $\mathrm{P}$ and ends at $\mathrm{Q}$, the equation for the fold curve is
\begin{equation}
 \mathrm{rot}(\mathbf{p}, \mathbf{a}_\theta, t) \,\,\,\, (0 \le t \le \psi),
  \label{Eq:newFoldCurve}
\end{equation}
where angle $\psi$ is calculated from the equation
\[
\psi=\cos^{-1} \left(\frac{(\mathbf{p}-(\mathbf{a}_\theta\cdot\mathbf{p})\mathbf{a}_\theta)\cdot(\mathbf{q}-(\mathbf{a}_\theta\cdot\mathbf{q})\mathbf{a}_\theta)}{1-(\mathbf{a}_\theta\cdot\mathbf{p})^2}\right).
\]
Here, $\mathbf{a}_\theta\cdot\mathbf{p} = \mathbf{a}_\theta\cdot\mathbf{q}$ because $\mathrm{P}$ and $\mathrm{Q}$ lie on the same small circle $(\mathbf{a}_\theta, \mathbf{p})$; thus the projected vectors $\mathbf{p}-(\mathbf{a}_\theta\cdot\mathbf{p})\mathbf{a}_\theta$ and $\mathbf{q}-(\mathbf{a}_\theta\cdot\mathbf{q})\mathbf{a}_\theta$ have equal norms $\sqrt{1-(\mathbf{a}_\theta\cdot\mathbf{p})^2}$, and dividing by this squared norm normalizes the expression so that the argument of $\cos^{-1}$ lies in $[-1,1]$.
The vectors $(\mathbf{a}_\theta\cdot\mathbf{p})\mathbf{a}_\theta$ and $(\mathbf{a}_\theta\cdot\mathbf{q})\mathbf{a}_\theta$ are identical and represent the position vector of the center of the small circle $\mathrm{O}_2$.
The position vector of $\mathrm{O}_1$ is obtained automatically as $2(\mathbf{a}_\theta\cdot\mathbf{p})\mathbf{a}_\theta$.
A new point $\mathrm{S}$ is set for the following operation. This point is the most distant point on the new fold curve from both $\mathrm{P}$ and $\mathrm{Q}$.

To fold the triangle $\mathrm{PQR}$ based on the new fold curve Eq.~\ref{Eq:newFoldCurve}, the position vector of the point $\mathrm{R}'$ must be calculated.
It can be obtained in two ways: one is the use of reflection and rotation, and the other is the use of reflection only.
For the former method, the point $\mathrm{R}$ is transferred to $\mathrm{R}''$ by the reflection across $\mathbf{a}_\mathrm{PQ}$. The transfer is carried out for not the triangle $\mathrm{PQR}$ but the spherical area $\mathrm{PRSQ}$.
The reflection is necessary to ensure that the whole area is not strained.
In other words, this operation is required to keep the original shape during the folding process.
Then, the area is rotated by $2\theta$ with respect to the axis $\mathrm{PQ}$ to overlap the reflected fold curve $\mathrm{PS'Q}$ with the fold curve $\mathrm{PSQ}$.
Therefore, the transformation from $\mathrm{R}$ to $\mathrm{R}'$ is described as
\begin{equation}
    \mathbf{r}' = \mathrm{rot}\left(\mathrm{ref}(\mathbf{r}, \mathbf{a}_\mathrm{PQ})-\mathbf{p},\frac{\mathbf{q}-\mathbf{p}}{\left|\mathbf{q}-\mathbf{p}\right|}, 2\theta\right)+\mathbf{p}.
    \label{Eq:3DFold1}
\end{equation}
For the latter method, the area $\mathrm{PRSQ}$ is transferred by reflection across the plane determined by $\mathrm{P}$, $\mathrm{Q}$, and $\mathrm{O}_2$.
Therefore, the equation for $\mathrm{R}'$ is
\begin{equation}
    \mathbf{r}' = \mathrm{ref}(\mathbf{r}-(\mathbf{a}_\theta\cdot\mathbf{p})\mathbf{a}_\theta, \mathbf{a}_\theta)+(\mathbf{a}_\theta\cdot\mathbf{p})\mathbf{a}_\theta.
    \label{Eq:3DFold2}
\end{equation}
The construction above assumes that $\mathrm{P}$ and $\mathrm{Q}$ are not antipodal. In the antipodal case, $\mathbf{a}_\mathrm{PQ}$ and $\mathbf{p}_\mathrm{m}$ are indefinite, so Eqs.~\ref{Eq:3DFold1} and~\ref{Eq:3DFold2} are not applicable; geodesics connecting antipodal points are not unique, and any such geodesic remains in $\mathbb{S}^2$ for arbitrary $\theta$. We therefore restrict attention to non-antipodal $\mathrm{P}, \mathrm{Q}$ in what follows.

The two transformations Eqs.~\ref{Eq:3DFold1} and~\ref{Eq:3DFold2} in fact define the same isometry, as the following theorem summarizes.

\begin{theorem}[3D Fold via Equidistant Curve]\label{thm:3DFold}
Let $\mathrm{P}, \mathrm{Q} \in \mathbb{S}^2$ be non-antipodal points, and let $\theta \in [-\theta_\mathrm{max}, \theta_\mathrm{max}]$, where $\theta_\mathrm{max}$ is determined by the inner angles of the spherical polygon being folded.
The small arc with pole $\mathbf{a}_\theta$ defined by Eq.~\ref{Eq:newFoldCurve}, joining $\mathrm{P}$ and $\mathrm{Q}$, realizes an isometric 3D fold of the spherical polygon onto the congruent polygon on the reflected sphere $\mathrm{O}_1$.
The induced transformation of a point $\mathbf{r}$ in the folded region admits the two equivalent expressions Eq.~\ref{Eq:3DFold1} (reflection followed by rotation) and Eq.~\ref{Eq:3DFold2} (a single planar reflection).
In particular, the special case $\theta = 0$ recovers the flat fold across the great arc $\mathrm{PQ}$ on $\mathbb{S}^2$ described in Section~\ref{subsec:DefAndAxioms}.
\end{theorem}

\begin{proof}
The construction of $\mathbf{a}_\theta$ together with Eq.~\ref{Eq:newFoldCurve} specifies the small arc through $\mathrm{P}$ and $\mathrm{Q}$ uniquely once $\theta$ is fixed.
Because $\mathrm{O}$ and $\mathrm{O}_1$ are unit spheres sharing this small circle, the reflection across the plane containing the circle maps $\mathrm{O}$ to $\mathrm{O}_1$ isometrically, so the polygon is transported onto $\mathrm{O}_1$ without distortion.

To see the equivalence of Eqs.~\ref{Eq:3DFold1} and~\ref{Eq:3DFold2}, note that the composition in Eq.~\ref{Eq:3DFold1} consists of (i) a reflection across the plane $\pi_1$ containing the great circle $\mathbf{a}_\mathrm{PQ}$ (with normal $\mathbf{a}_\mathrm{PQ}$), followed by (ii) a rotation by $2\theta$ about the chord $\mathrm{PQ}$, which lies in $\pi_1$.
By a classical result in geometry, the composition of a reflection across a plane $\pi_1$ and a rotation by $2\theta$ about a line $l$ contained in $\pi_1$ equals the reflection across the plane $\pi_2$ that contains $l$ and makes angle $\theta$ with $\pi_1$.
In our case, $l = \mathrm{PQ}$ and $\pi_2$ is the plane perpendicular to $\mathbf{a}_\theta$ (which contains $\mathrm{P}$ and $\mathrm{Q}$ since both lie on the small circle with pole $\mathbf{a}_\theta$), and the angle between $\pi_1$ and $\pi_2$ equals $\theta$ by the definition of $\mathbf{a}_\theta$.
Equation~\ref{Eq:3DFold2} is precisely this reflection across $\pi_2$, as can be verified by expanding $\mathrm{ref}(\mathbf{r}-(\mathbf{a}_\theta\cdot\mathbf{p})\mathbf{a}_\theta, \mathbf{a}_\theta)+(\mathbf{a}_\theta\cdot\mathbf{p})\mathbf{a}_\theta = \mathbf{r} - 2(\mathbf{a}_\theta\cdot\mathbf{r} - \mathbf{a}_\theta\cdot\mathbf{p})\mathbf{a}_\theta$, which is the reflection of $\mathbf{r}$ across the affine plane $\{\mathbf{x} : \mathbf{a}_\theta\cdot\mathbf{x} = \mathbf{a}_\theta\cdot\mathbf{p}\}$.
Setting $\theta=0$ gives $\mathbf{a}_\theta=\mathbf{a}_\mathrm{PQ}$, $\pi_2=\pi_1$, and Eq.~\ref{Eq:3DFold2} reduces to the planar reflection Eq.~\ref{Eq:Reflection}.
\end{proof}

In the computer graphics below, Eq.~\ref{Eq:3DFold1} is used because it is more straightforward to implement.

\begin{figure}[ht]
\centering
\includegraphics[width=0.95\textwidth]{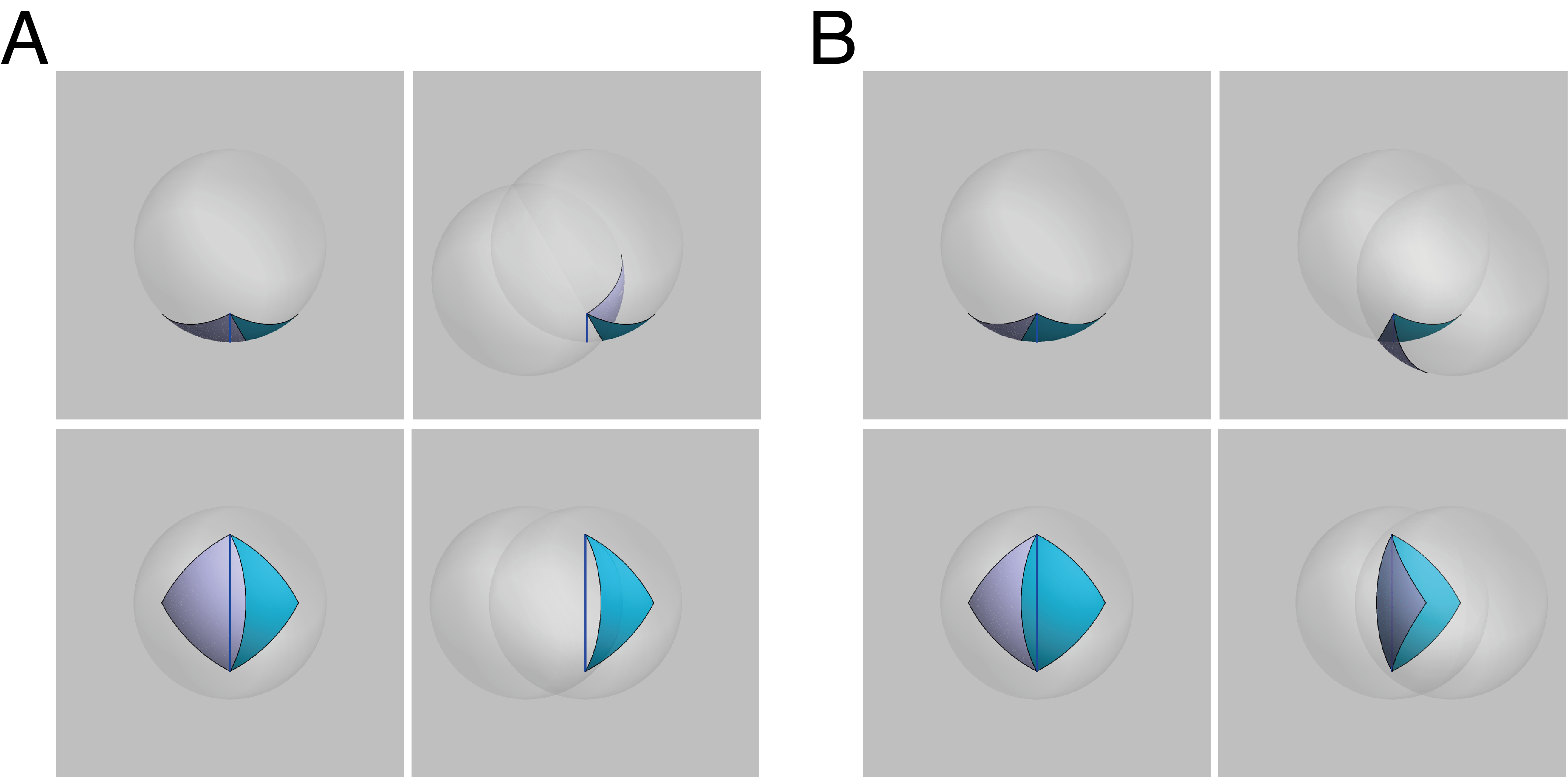}
\caption{Two types of folding. The left spherical area in gray is folded in both cases. A: Folding to the inside of the original sphere. B: Folding to the outside. This figure first appeared in~\cite{yoshinoKatachi2024}.}\label{Fig:3DFoldExample}
\end{figure}

The geometrical difference in folding directions should be mentioned because it determines the location of the new fold curve.
Equation~\ref{Eq:newFoldCurve} holds for both cases: $-\theta_\mathrm{max}\le\theta<0$ for the outside and $0<\theta\le \theta_\mathrm{max}$ for the inside, respectively. The value of the positive parameter $\theta_\mathrm{max}$ is determined by its foldout diagram.
Figure~\ref{Fig:3DFoldExample} shows an example of each case; Fig.~\ref{Fig:3DFoldExample}A is the fold to the inside of the sphere, and Fig.~\ref{Fig:3DFoldExample}B is that to the outside.
The top and bottom images give top and side views, respectively.
The left and right images for groups A and B represent the images before and after folding.
Thick blue lines represent the original fold curves (great arcs), and the boundaries of different colored areas represent the new fold curves.
The figure includes the sphere containing the folded area to indicate the relationship between the two areas before and after the foldings.

Because the new fold curves are not the same as the original fold curves, the area to be folded is different from the original spherical polygon. Significantly, the fold curve to the inside is drawn in the area outside the original polygon as shown in Figs.~\ref{Fig:NewFoldCurve}B~and~\ref{Fig:3DFoldExample}A.
This feature of folding to the inside sometimes prevents folding designs. It is mentioned again using an example in the next subsection.

The dihedral angle produced by the 3D fold of Theorem~\ref{thm:3DFold} can be expressed in closed form.

\begin{proposition}[Dihedral Angle of the 3D Fold]\label{prop:DihedralAngle}
Under the setting of Theorem~\ref{thm:3DFold}, the angle $\alpha$ between the tangential planes of the spherical surfaces $\mathrm{O}$ and $\mathrm{O}_1$ at the apex $\mathrm{S}$ of the new fold curve (Fig.~\ref{Fig:NewFoldCurve}) satisfies
\begin{equation}
\alpha = 2 \cos^{-1} \left(\sqrt{\frac{1+\mathbf{p}\cdot\mathbf{q}}{2}}\sin\theta\right).
\label{Eq:GeneralRelation}
\end{equation}
In particular, $\alpha$ depends on both the spherical distance between $\mathrm{P}$ and $\mathrm{Q}$ (through $\mathbf{p}\cdot\mathbf{q}$) and the fold parameter $\theta$.
\end{proposition}

\begin{proof}
Since $\mathbf{a}_\mathrm{PQ}\cdot\mathbf{p}_\mathrm{m}=0$, the rotation formula for $\mathbf{a}_\theta$ simplifies to
\[
\mathbf{a}_\theta = \mathbf{a}_\mathrm{PQ}\cos\theta + \mathbf{p}_\mathrm{m}\sin\theta.
\]
Because $\mathbf{a}_\mathrm{PQ}\cdot\mathbf{p}=0$ (P lies on $\mathbf{a}_\mathrm{PQ}$) and $\mathbf{p}_\mathrm{m}\cdot\mathbf{p} = \sqrt{(1+\mathbf{p}\cdot\mathbf{q})/2}$, it follows that
\[
\mathbf{a}_\theta\cdot\mathbf{p} = \sqrt{\frac{1+\mathbf{p}\cdot\mathbf{q}}{2}}\sin\theta.
\]
The outward unit normal to $\mathrm{O}_1$ at $\mathrm{S}$ is $\mathbf{n}_1 = \mathbf{s} - 2(\mathbf{a}_\theta\cdot\mathbf{s})\mathbf{a}_\theta$ (with $|\mathbf{n}_1|=1$ verified from $\mathbf{O}_1 = 2(\mathbf{a}_\theta\cdot\mathbf{p})\mathbf{a}_\theta$), and $\mathbf{a}_\theta\cdot\mathbf{s} = \mathbf{a}_\theta\cdot\mathbf{p}$ since $\mathrm{S}$ lies on the same small circle.
The angle $\alpha$ is defined as the supplement of the angle between the outward normals $\mathbf{s}$ and $\mathbf{n}_1$:
\[
\cos\alpha = -\mathbf{s}\cdot\mathbf{n}_1 = 2(\mathbf{a}_\theta\cdot\mathbf{p})^2 - 1.
\]
Applying $\cos\alpha = 2\cos^2(\alpha/2)-1$ gives $\cos(\alpha/2) = |\mathbf{a}_\theta\cdot\mathbf{p}|$, which yields Eq.~\ref{Eq:GeneralRelation}.
\end{proof}

Another constraint for 3D folding is that the new fold curve must be drawn within the face to be drawn.
The new fold curve is drawn on the spherical polygon to be folded (fold to the outside) or the one opposite to the polygon (fold to the inside) so the curve cannot exceed its boundary.
This means that the fold angle is restricted by the geometry of the spherical polygon on which the new fold curve is drawn. For example, the angle $\theta$ in Fig.~\ref{Fig:NewFoldCurve} is expressed in terms of $\theta_\mathrm{max}=\mathrm{min}(\angle\mathrm{P},\angle\mathrm{Q})$ as $-\theta_\mathrm{max}\le\theta\le\theta_\mathrm{max}$, where $\angle\mathrm{P}$ and $\angle\mathrm{Q}$ represent the amounts of the inner angle of the polygon.

\subsection{Spherical Origami Birds in the 3D Space} \label{sec:wings}

The expansion of the wings of a spherical origami bird is subject to two constraints that do not arise in the Euclidean case.
First, folding the inside wing to the inside of the sphere is impossible using the fold curve in Fig.~\ref{Fig:FoldoutAndResultsInS2} directly, because the required new fold curves would be drawn on the neck and tail parts, which obstruct the fold.
Second, folding both wings to the outside avoids this obstruction but is unacceptable as a ``bird'' form, since both wings expand on the same side.
Therefore, the foldout diagram must be modified to satisfy both constraints simultaneously.
Recall from Section~\ref{sec:3D} that the new fold curve for an inside fold must be drawn on the face opposite the folded polygon.
In the original crane design, this opposite face coincides with the neck and tail regions, making an inside fold of the inner wing impossible without modification.
To resolve this, the inner-wing face in the foldout diagram is extended along the fold edge to include a small additional region (visible as the blue-bounded triangles in the top-right panels of Fig.~\ref{Fig:3DFoldout}).
The new fold curve is drawn within this extension, which lies entirely within the wing and does not overlap the neck or tail.
The outer wing retains its original face area, since an outside fold ($\theta < 0$) requires no such extension.

\begin{figure}[ht]
\centering
\includegraphics[width=0.95\textwidth]{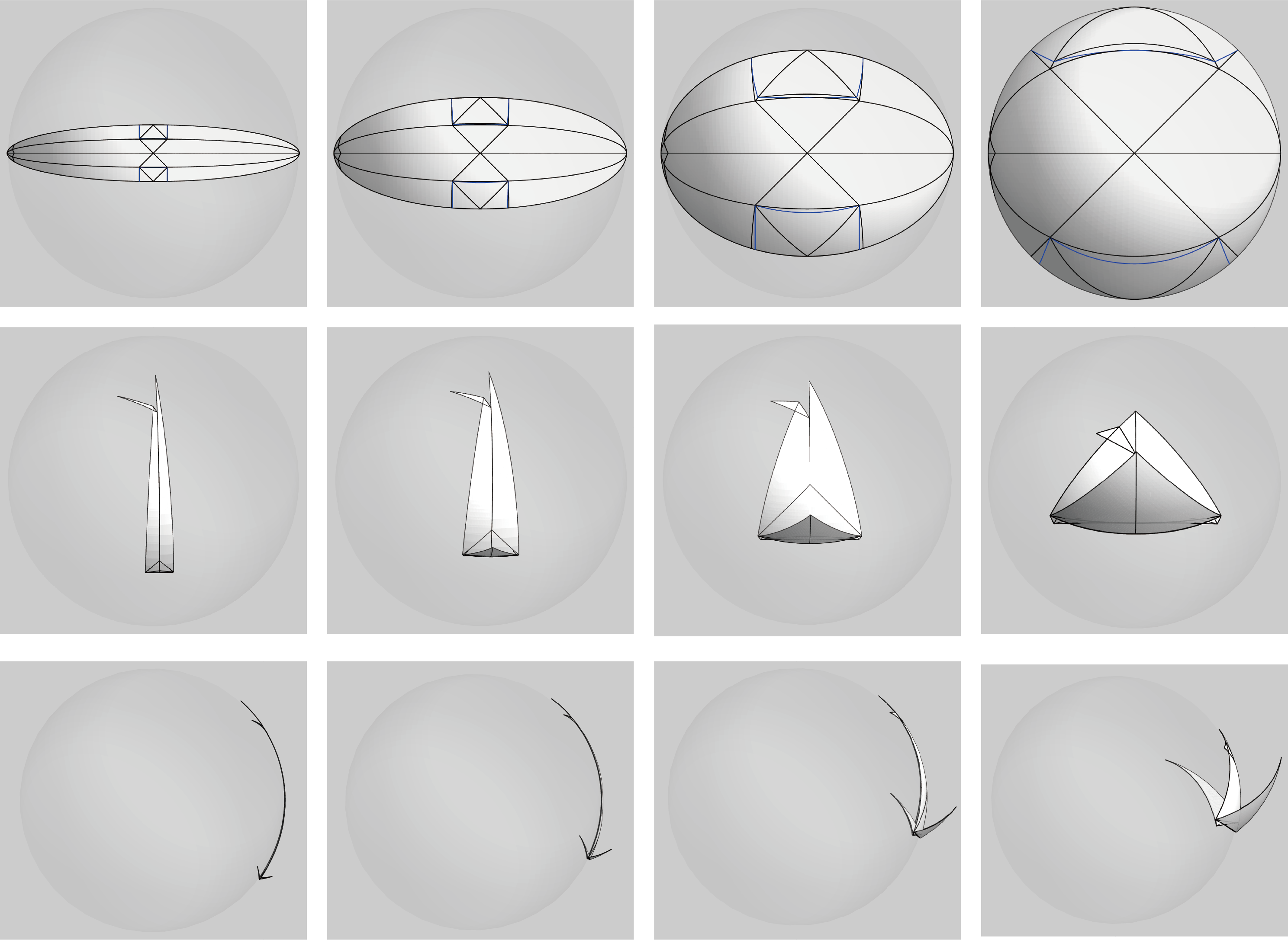}
\caption{New foldout diagrams (top row) and their folds (middle and bottom rows). The blue curves in the top-right panels show the new fold curves: on the inner wing, the fold curve lies within a small extension of the wing face; on the outer wing, it lies on the original wing face.}\label{Fig:3DFoldout}
\end{figure}

Figure~\ref{Fig:3DFoldout} shows alternative foldout diagrams (top row), and the side and front views of the resulting folds (middle and bottom rows, respectively).
The diagrams are designed so that each wing is extended oppositely.
The inner wing uses an inside fold ($\theta > 0$) with the new fold curve drawn in the extended region, while the outer wing uses an outside fold ($\theta < 0$) with the new fold curve drawn on the original wing face.
Both wings use the same fold-angle magnitude $|\theta|$ but opposite signs, so that Eq.~\ref{Eq:GeneralRelation} yields identical wing lengths regardless of folding direction.
The new fold curve is defined to contact the original fold curve, ensuring that the distances between the wing tip and wing bottom are equal for both wings.
The blue curves represent the new fold curves added to the original foldout diagrams for the 3D foldings.
The folds were finished successively in all the cases.

The results indicate that the equations for the new fold curves (Eq.~\ref{Eq:newFoldCurve}) and transformation of points (Eqs.~\ref{Eq:3DFold1} and~\ref{Eq:3DFold2}) work correctly.
The images of the folds demonstrate that the angles between two wings depend on the shape of the spherical sheet. The length of the original fold curve results in the difference as $\mathbf{p}\cdot\mathbf{q}$ in Eq.~\ref{Eq:GeneralRelation}.

\section{Discussion}\label{sec:Discussion}

The results above place spherical origami on the same axiomatic and flat-foldability footing as Euclidean origami, while introducing a one-parameter extension to 3D folding via equidistant curves.
Several questions about scope and applicability remain; we discuss them in turn.

It should be noted that Theorem~\ref{thm:SphericalKJ} concerns flat folds (those remaining on $\mathbb{S}^2$); whether an analogous condition governs the 3D folds of Theorem~\ref{thm:3DFold} is an open question, since the fold angles there depend on the global geometry of the face rather than only the vertex angles.
Furthermore, while Theorem~\ref{thm:SphericalKJ} is rigorously established, the practical use of the framework has been demonstrated only on highly symmetric patterns (origami birds); general multi-vertex crease patterns may exhibit additional global constraints beyond the local vertex condition that are not addressed here.

To our knowledge, this paper is the first to present the equations for 3D folds of spherical sheets.
The equations for fold curves and transformations of spherical points provide the necessary computational basis for designing folds, obtaining foldout diagrams, and producing computer graphics of spherical origami.
Furthermore, the results suggest that equidistant curves (small circles on $\mathbb{S}^2$) play a role in spherical origami analogous to that of hypercycles in the hyperbolic setting, and their incorporation may be necessary for a complete treatment of non-Euclidean origami---a direction not yet explored in previous reports such as Alperin et al.~\cite{hyperbolicOrigami} for hyperbolic origami.

Although the present study addresses only one practical case, spherical birds, the framework is in principle applicable to any crease pattern whose sheet satisfies the incircle condition for origami birds on $\mathbb{S}^2$.
For example, other traditional designs such as a spherical helmet (\textit{kabuto}) or boat could be treated by the same procedure: draw the foldout diagram using the seven axioms, verify the angle condition of Theorem~\ref{thm:SphericalKJ} at each vertex, and compute the 3D fold curves using Eq.~\ref{Eq:newFoldCurve}.
The main obstacle to a general treatment is that the fold angles for 3D folds must be determined globally at the diagram stage, a constraint that becomes increasingly complex as the number of vertices grows.
Further practical applications are therefore subject to specific constraints that do not arise in Euclidean origami: most notably, the angles for the 3D folds must be determined at the time of drawing the foldout diagrams.
(1) A fold to the inside of the sphere is sometimes impossible due to the nature of the fold curves outside the face, as demonstrated in the present study.
(2) The fold angle is restricted by the shape of the face to be drawn.

The approach demonstrated in Section~\ref{sec:wings}---inner wing folded inward via a small extension of the wing face, outer wing folded outward---is one of several possible designs. An alternative is to fold the neck and body to the outside of the sphere while keeping the inside wing fixed; this requires additional fold curves and differs from the usual origami-bird folding procedure.

\section{Concluding Remarks}

This paper has established a geometrical framework for spherical origami in two settings: on $\mathbb{S}^2$ and in 3D space.
Theorem~\ref{thm:SphericalHJ} shows that all seven Huzita--Justin axioms admit explicit constructions on $\mathbb{S}^2$, and Theorem~\ref{thm:SphericalKJ}, together with Corollary~\ref{cor:Ftiling}, establishes the Kawasaki--Justin flat-foldability condition and its equivalence with Robertson's $f$-tiling vertex condition.
For 3D folding, Theorem~\ref{thm:3DFold} introduces equidistant curves as fold curves and yields a one-parameter family of isometric folds that recover great-arc folds in the flat limit, with the dihedral angle given in closed form by Proposition~\ref{prop:DihedralAngle}.
Spherical origami birds are constructed in 3D space as a practical demonstration.
These results provide a foundation for further investigation of non-Euclidean origami, including the hyperbolic case.

\section*{Conflict of interest}
The author declares that he has no conflicts of interest concerning this article.

\section*{Data availability}
Data sharing not applicable to this article as no datasets were generated or analysed during the current study.

\section*{Acknowledgements}
The author thanks Prof.\ K.\ Kaino of the National Institute of Technology, Sendai College for his advice at the early stage of this study, and Drs.\ S.\ Chaidee and N.\ Ploymaklam of Chiang Mai University for helpful discussions. Preliminary results were presented at JCDCG$^3$ 2024 and at the 97th Symposium of the Society for Science on Form~\cite{yoshinoJCDCG2024,yoshinoKatachi2024}. This research was partially supported by the Inoue Enryo Memorial Research Grant of Toyo University.

\section*{Use of generative AI}
During the preparation of this work the author used Claude (Anthropic) to assist with restructuring the manuscript into theorem--proof format, \LaTeX{} editing, and language refinement. The author reviewed and edited the content and takes full responsibility for the content of this preprint.

\appendix

\section{Definitions and Axioms for Euclidean Origami} \label{Appendix:PlanarOrigami}
The following is quoted from Alperin and Lang~\cite{alperinAxioms}.

\subsection{Definitions for Euclidean Origami} \label{subsec:DefinitionsPlanar}

\begin{itemize}
    \item Definition 1(Point): A point ($x$, $y$) is the ordered pair where $x$ and $y$ are the Cartesian coordinates of the point.
    \item Definition 2 (Line): A line ($X$, $Y$) is the set of points ($x$,$y$) that satisfy the equation $Xx+Yy+1=0$.
    \item Definition 3 (Folded Point): The \textit{folded image} $F_{L_F}(P)$ of a point $P=(x,y)$ in a fold line $L_F=(X_F, Y_F)$ is the reflection of the point in the fold line.
    \item Definition 4 (Folded Line): The \textit{folded image} $F_{L_F}(L)$ of a line $L=(X, Y)$ in a fold line $L_F=(X_F, Y_F)$ is the reflection of the line in the fold line.
    \item Definition 5 (Point-Point Alignment): Given two points $P_1=(x_1, y_1)$ and a point $P_2=(x_2, y_2)$, the alignment $P_1\leftrightarrow P_2$ is satisfied if and only if $x_1=x_2$ and $y_1=y_2$.
    \item Definition 6 (Line-Line Alignment): Given two lines $L_1=(X_1, Y_1)$ and another line $L_2=(X_2, Y_2)$, the alignment $L_1\leftrightarrow L_2$ is satisfied if and only if $X_1=X_2$ and $Y_1=Y_2$.
    \item Definition 7 (Point-Line Alignment): Given a point $P=(x,y)$ and a line $L=(X, Y)$, the alignment $P\leftrightarrow L$ is satisfied if and only if $Xx+Yy+1=0$.
    \item Definition 8 (One-Fold Axiom): A \textit{one-fold axiom} (1FA) is a minimal set of alignments that define a single fold line on a finite region of the Euclidean plane with a finite number of solutions.
\end{itemize}

\subsection{Huzita--Justin Axioms} \label{subsec:HuzitaJustin}
\begin{itemize}
    \item Axiom 1: Given two points $P_1$ and $P_2$, we can fold a line connecting them.
    \item Axiom 2: Given two points $P_1$ and $P_2$, we can fold $P_1$ onto $P_2$.
    \item Axiom 3: Given two lines $L_1$ and $L_2$, we can fold line $L_1$ onto $L_2$.
    \item Axiom 4: Given a point $P_1$ and a line $L_1$, we can make a fold perpendicular to $L_1$ passing through the point $P_1$.
    \item Axiom 5: Given two points $P_1$ and $P_2$ and a line $L_1$, we can make a fold that places $P_1$ onto $L_1$ and passes through the point $P_2$.
    \item Axiom 6: Given two points $P_1$ and $P_2$ and two lines $L_1$ and $L_2$, we can make a fold that places $P_1$ onto line $L_1$ and places $P_2$ onto line $L_2$.
    \item Axiom 7: Given a point $P_1$ and two lines $L_1$ and $L_2$, we can make a fold perpendicular to $L_2$ that places $P_1$ onto line $L_1$.
\end{itemize}

\bibliographystyle{elsarticle-num}
\bibliography{_bibliography.bib}

\end{document}